\newif\ifprocs
\procstrue
\procsfalse 

\ifprocs
\documentclass[twoside,leqno,twocolumn]{article}
\usepackage{ltexpprt}
\else
\documentclass[11pt]{article}
\fi

\ifprocs
\usepackage[fleqn]{amsmath}
\else
\usepackage{amsthm}
\usepackage{amsmath}
\fi
\usepackage{amssymb,sectsty,url}
\usepackage[letterpaper,hmargin=1.0in,vmargin=1.0in]{geometry}

\usepackage{ifpdf}
\ifpdf    
\usepackage[pdftex,colorlinks,linkcolor=blue,citecolor=blue,filecolor=blue,urlcolor=blue]{hyperref}
\else    
\usepackage[hypertex]{hyperref}
\fi

\usepackage{cleveref}
\usepackage{color}
\usepackage[boxed]{algorithm}
\usepackage{tikz}
\usepackage{graphicx}
\usepackage{caption}
\usepackage{subcaption}
\usepackage{nicefrac}
\usepackage{aliascnt}
\usepackage[textsize=scriptsize]{todonotes}

\ifprocs
\crefname{@theorem}{Theorem}{Theorems}
\crefname{lemma}{Lemma}{Lemmas}
\crefname{proposition}{Proposition}{Propositions}
\crefname{corollary}{Corollary}{Corollaries}
\crefname{fact}{Fact}{Facts}

\newtheorem{definition}{Definition}[section]
\crefname{definition}{Definition}{Definitions}

\newtheorem{observation}{Observation}[section]
\crefname{observation}{Observation}{Observations}

\newtheorem{claim}{Claim}[section]
\crefname{claim}{Claim}{Claims}

\else

\newtheorem{theorem}{Theorem}[section]
\crefname{theorem}{Theorem}{Theorems}

\newaliascnt{lemma}{theorem}
\newtheorem{lemma}[lemma]{Lemma}
\aliascntresetthe{lemma}
\crefname{lemma}{Lemma}{Lemmas}

\newaliascnt{proposition}{theorem}

\aliascntresetthe{proposition}
\crefname{proposition}{Proposition}{Propositions}

\newaliascnt{corollary}{theorem}
\newtheorem{corollary}[corollary]{Corollary}
\aliascntresetthe{corollary}
\crefname{corollary}{Corollary}{Corollaries}

\newaliascnt{fact}{theorem}
\newtheorem{fact}[fact]{Fact}
\aliascntresetthe{fact}
\crefname{fact}{Fact}{Facts}

\newaliascnt{definition}{theorem}
\newtheorem{definition}[definition]{Definition}
\aliascntresetthe{definition}
\crefname{definition}{Definition}{Definitions}

\newaliascnt{remark}{theorem}

\aliascntresetthe{remark}
\crefname{remark}{Remark}{Remarks}

\newaliascnt{conjecture}{theorem}

\aliascntresetthe{conjecture}
\crefname{conjecture}{Conjecture}{Conjectures}

\newaliascnt{claim}{theorem}
\newtheorem{claim}[claim]{Claim}
\aliascntresetthe{claim}
\crefname{claim}{Claim}{Claims}

\newaliascnt{question}{theorem}

\aliascntresetthe{question}
\crefname{question}{Question}{Questions}

\newaliascnt{exercise}{theorem}

\aliascntresetthe{exercise}
\crefname{exercise}{Exercise}{Exercises}

\newaliascnt{example}{theorem}

\aliascntresetthe{example}
\crefname{example}{Example}{Examples}

\newaliascnt{notation}{theorem}

\aliascntresetthe{notation}
\crefname{notation}{Notation}{Notations}

\newaliascnt{problem}{theorem}

\aliascntresetthe{problem}
\crefname{problem}{Problem}{Problems}

\newaliascnt{construction}{theorem}

\aliascntresetthe{construction}
\crefname{construction}{Construction}{Constructions}

\newaliascnt{observation}{theorem}
\newtheorem{observation}[observation]{Observation}
\aliascntresetthe{observation}
\crefname{observation}{Observation}{Observations}
\fi

\ifprocs
\newcommand{\qedsymbol}{$\square$}
\fi

\newcommand{\norm}[1]{\lVert#1\rVert}

\newcommand{\R}{\mathbb R}
\newcommand{\Z}{\mathbb Z}

\newcommand{\B}{\{ 0,1 \}}

\begin{document}

\ifprocs
\else
\begin{titlepage}
\fi

\title{Near-Optimal (Euclidean) Metric Compression}
\date{}

\ifprocs
\author{
    Piotr Indyk\thanks{CSAIL MIT.}
  \and
    Tal Wagner\thanks{CSAIL MIT.}
}
\else
\author{
    Piotr Indyk\footnotemark[\value{footnote}]\thanks{\texttt{indyk@mit.edu}.}\\
    MIT
  \and
    Tal Wagner\thanks{\texttt{talw@mit.edu}.}\\
    MIT
}
\clearpage
\fi

\maketitle

\ifprocs
\else
\thispagestyle{empty}
\fi

\begin{abstract}
The metric sketching problem is defined as follows. Given a metric on $n$  points, and $\epsilon>0$, we wish to produce a small size data structure (sketch) that, given any pair of point indices, recovers the distance between the points up to a $1+\epsilon$ distortion. In this paper we consider metrics induced by $\ell_2$ and $\ell_1$ norms whose spread (the ratio of the diameter to the closest pair distance) is bounded by $\Phi>0$. A well-known dimensionality reduction theorem due to Johnson and Lindenstrauss yields a sketch of size $O(\epsilon^{-2} \log (\Phi n) n\log n)$, i.e., 
 $O(\epsilon^{-2}  \log (\Phi n) \log n)$ bits per point. We show that this bound is not optimal, and can be substantially improved to 
$O(\epsilon^{-2}\log(1/\epsilon) \cdot \log n + \log\log \Phi)$ bits per point. Furthermore, we show that our bound is tight up to a factor of $\log(1/\epsilon)$. 

We also consider sketching of general metrics and provide a sketch of size $O(n\log(1/\epsilon)+ \log\log \Phi)$ bits per point, which we show is optimal. 

\end{abstract}

\ifprocs
\else
\end{titlepage}
\fi

\section{Introduction}

Compact representations (or sketches) of high-dimensional data are very useful tools for applications involving data storage, processing and analysis. A prototypical example of this approach is the Johnson-Lindenstrauss theorem~\cite{johnson1984extensions}, which states that any set of $n$ points in the Euclidean space of arbitrary dimension can be mapped into a space of dimension $O(\epsilon^{-2}\log n)$ such that the distances between any pair of points are preserved up to a factor of $1+\epsilon$. The theorem makes it possible to represent each point as a sequence of  $O(\epsilon^{-2}\log n)$ {\em numbers}. The number of {\em bits} needed to represent a point depends on the precision of the underlying data set. If the coordinates of the high-dimensional points come from the range $\{-\Phi \ldots \Phi\}$, and if one uses the ``binary'' variant of the Johnson-Lindenstrauss theorem due to~\cite{achlioptas2003}, then each coordinate of  the projected points can be represented using $O(\log (n\Phi))$ bits. This yields $O( \epsilon^{-2} \log (n) \log (n\Phi))$ bits per point. 

Perhaps surprisingly, it is not known whether this bound is tight
for metric sketching,
or whether the sketch length can be improved much further. This could be in part because the Johnson-Lindenstrauss theorem itself was not known to be optimal until very recently~\cite{LarsenN16} (although it was known to be optimal up to a factor of $\log(1/\epsilon)$~\cite{alon2003jllowerbound}). Still, even an optimal lower bound for dimensionality reduction does not imply a corresponding bound on the sketch length. Another set of related results are the lower bounds from~\cite{jayram2013optimal,molinaro2013beating}, which show that 
$\Omega( \epsilon^{-2} n \log (n/\delta) \log \Phi)$ bits are needed to sketch the $\ell_1$ or $\ell_2$ distances between two sets of $n$ points with probability $1-\delta$. The latter lower bounds match the ``discretized Johnson-Lindenstrauss'' upper bound outlined above. 
However, they apply only to the {\em one-way communication} variant of the problem, where each point is held by either Alice or Bob, who  communicate in order to estimate the distances across the partition. In contrast, in our setting, a single party holds and compresses the {\em whole} data set (see \Cref{sec:preliminaries} for the formal definition). Thus the lower bounds of \cite{jayram2013optimal,molinaro2013beating} are also inapplicable in our context. 

\paragraph{Our results.} In this paper we show that the ``discretized Johnson-Lindenstrauss'' sketching bound is {\em not} tight. In particular, we describe a new randomized sketching algorithm for $n$ points that enables estimating the distances up to a factor of $1 +\epsilon$ using at most $O(\epsilon^{-2}\log(1/\epsilon) \log n + \log\log \Phi)$ bits per point. This substantially improves  over the earlier bound, replacing the $\log (n\Phi)$ term by $\log(1/\epsilon)$, and exponentially reducing the dependence on the precision parameter $\Phi$. Furthermore, we show a lower bound of  $O(\epsilon^{-2}  \log n + \log\log \Phi)$, which implies that the bounds are tight  up to a factor of  $\log(1/\epsilon)$. This result, as well as all results that follow, extend to the setting where the coordinates of the points are arbitrary real numbers but the spread of the pointset (the ratio of its diameter to the closest pair distance) is at most $\Phi$.

The result for the Euclidean metric is a corollary of a more general result about sketching of {\em arbitrary} $\ell_p$ norms.  Specifically, we show that any metric induced by a set of $n$ points in a $d$-dimensional space under an $\ell_p$ norm can be sketched using $O((d+\log n)\log(1/\epsilon) + \log\log \Phi)$ bits per point. 
The result for the Euclidean norm is then obtained by letting $d=O(\epsilon^{-2} \log n)$. Furthermore, since the $\ell_\infty$ norm is universal (any $n$-point metric space can be embedded into $\ell_{\infty}^n$), it also follows that any $n$-point metric can be sketched using at most $O(n \log(1/\epsilon) + \log\log \Phi)$ bits per point,
which we show is tight.
This improves over a naive bound of $O(n \log (\log (\Phi)/\epsilon) )$ obtained by rounding each distance to the nearest integer power of $1+\epsilon$ and storing the exponent, after an appropriate scaling. 

\paragraph{Related Work.} Distance-preserving sketches and data structures have been studied extensively. In what follows we focus on the prior work that is most relevant to the results in this paper. 

As described earlier, it is known that one can approximate the distances between $n$ points in the Euclidean space using $O( \epsilon^{-2} \log (n) \log (n\Phi))$ bits per point. If the goal is to only preserve distances in a certain range, say in the interval $[t, 10t]$ for some parameter $t$, then one can achieve a $O(\epsilon^{-2} \log (n))$ bits bound using distance sketches due to~\cite{KOR98}.
However, in general, there are $O(\log \Phi)$ different ``scales'' $t$ to preserve, so this approach does not improve over the aforementioned bound.
Similarly, the very recent work~\cite{AlonK16} focused on approximating Euclidean distances between points of norm at most $1$, up to an \emph{additive} error of $\epsilon$, and showed a sketch of size $O(\epsilon^{-2}\log n)$ bits per point.

Distance labeling is a general approach to the class of the problems considered in this paper. In particular, it is known~\cite{peleg1989graph,thorup2005approximate} that any metric over $n$ nodes can be represented by roughly $n^{1+\Theta(1/c)}$ numbers while distorting the distances by a factor of $c$. Note that in order to achieve a near-linear sketch size, the distortion must be almost logarithmic.

Quadtrees are  simple and popular data structures for storing point sets (see~\cite{samet1988overview}), often used in practice for low dimensional data, typically 2D or 3D. They are further related to our algorithm, as both are based on a hierarchical clustering of the points. For arbitrary low dimensional points, the quadtree size can still be as large as $n\log\Phi$. Several works~\cite{deBernardo2013,VenkatM14,Gagie15} showed that if the points are ``well clustered'' then the quadtree can be compressed into $O(1)$ bits per point, implying a similar bound for storing the metric space. \cite{hudson2009succinct} considers a different assumption on the point set, called well-seperatedness. They show that this assumption implies that the quadtree has at most $O(n)$ nodes,  which makes it possible  to store the metric with distortion $1+\epsilon$ using $O(\log(1/\epsilon)+(\log\Phi)/n)$ bits per point. Without structural assumptions on the point set, their techniques still require $\Omega(\log\Phi)$ bits per point.
\paragraph{Our techniques} 
Our construction is based on a hierarchical clustering of the metric space, which naturally forms a tree $T$ of clusters. The clustering at level $\ell$ is the partition to connected components induced by drawing edges between points at distance at most $2^\ell$. Each cluster is assigned a representative point called a \emph{center}. Note that unlike other typical variants of hierarchical clustering (eg.~quad-trees), the diameter of our clusters is unbounded in terms of $\ell$.

The tree size is first reduced to linear by compressing long paths of nodes with only one child. From a distance estimation point of view, this means that if a cluster is very well separated from the rest of the metric (in terms of the ratio between its diameter to the distance to the closest external point), then we can replace it entirely with its center for the purpose of estimating the distances between internal and external points.

After the compression, our aim is to store every point as the displacement from a nearby point that is already stored. To this end, we keep track of the structure of the clusters beyond what is given in $T$. A cluster $C$ at level $\ell$ is formed by drawing edges of length up to $2^\ell$ between some clusters $C_1,\ldots,C_k$ in level $\ell-1$. We fix a rooted spanning tree in this graph of clusters, and for every non-root cluster $C_j$, we store its center as the displacement from the closest point in $C_i$, its parent cluster in the spanning tree. We call that point the \emph{ingress} of $C_j$. The distance between the ingress and the center is bounded by $2^\ell+\mathrm{diam}(C_j)$, and does not depend on $\mathrm{diam}(C_i)$. This will ensure we pay (in storage cost) for the diameter of each cluster only once.

The displacement is rounded to a precision that depends on the diameter and the level, roughly $2^{-\ell}\mathrm{diam}(C_j)$, and we show that the total precision cost over all the clusters is linear. By adding the rounded displacement back to the ingress, we obtain an approximation for the center of $C_j$, which we call a \emph{surrogate}. Our estimate for the distance between two points would be the distance between their surrogates.

The above description is oversimplified, and our actual choice of ingresses and surrogates is more careful. First, the closest point to $C_j$ in $C_i$ may have been lost in the preceding compression, so we might need to settle for another nearby point. Second, we need to store the displacement not from the ingress itself, but rather from the surrogate of the ingress, since the latter is what will actually be available to us during estimation. This means the surrogates are defined recursively, and we need to prevent error from accummulating. Ultimately we show that within given parts of the tree (called \emph{subtrees}), we can recover the surrogates up to a fixed (unknown) shift, and this gives us satisfactory estimates for the original distances.

\section{Preliminaries and Formal Statements}\label{sec:preliminaries}
For an integer $n$, let $X$ be a fixed set of $n$ labeled points. Throughout we will use the convention $X=\{1,\ldots,n\}$. In the metric sketching problem, our goal is to design a sketch from which the distance between any pair of points can be approximately recovered, given their labels. 
Formally, let $\mathcal D_X$ be a family of metrics on $X$. For an integer $b>0$, we define a \emph{$b$-bit sketching scheme} for $\mathcal D_X$ as a pair of algorithms $(\mathsf{Summ},\mathsf{Est})$ such that
\begin{itemize}
  \item $\mathsf{Summ}$ is a (possibly randomized) summary algorithm, mapping a metric $D\in\mathcal D_X$ to a bitstring of length $b$.
  \item $\mathsf{Est}$ is an estimation algorithm, which takes the output of $\mathsf{Summ}$ on a metric $D\in\mathcal{D}_X$, and a pair of labels $x,y\in X$, and outputs an estimate for $D(x,y)$.
\end{itemize}
The sketching scheme has \emph{distortion $k$} if for every $D\in\mathcal D_X$, if $S_D$ is the output of a successful execution of $\mathsf{Summ}$ on $D$, then for every $x,y\in X$,
\[ D(x,y) \leq \mathsf{Est}(S_D,x,y) \leq k\cdot D(x,y). \]
\ifprocs
We denote by $b_k(\mathcal D_X)$ the infimum $b$ such that $\mathcal D_X$ has a $b$-bit sketching scheme with distortion $k$.
\else
We denote,
\[ b_k(\mathcal D_X) := \inf\{b : \text{$\mathcal D_X$ has a $b$-bit sketching scheme with distortion $k$} \}. \]
\fi
We are interested in upper bounds on $b_{1+\epsilon}(\mathcal D_X)$ for commonly arising metric families $\mathcal D_X$, and most importantly Euclidean metrics.

To state our results, we recall some basic definitions. The \emph{spread} of a metric $D$ on $X$ is the ratio $\frac{\max_{x,y\in X}D(x,y)}{\min_{x,y\in X, x\neq y}D(x,y)}$. For $p\geq1$, the \emph{$\ell_p$-norm} of a point $v\in\R^d$ is $\norm{v}_p:=\left(\sum_{i=1}^d|v_i|^p\right)^{1/p}$. The \emph{$\ell_\infty$-norm} is $\norm{v}_\infty:=\max_i|v_i|$. For $1\leq p\leq\infty$, a metric $D$ on $X$ is called a \emph{$d$-dimensional $\ell_p$-metric} if there is a map $f:X\rightarrow\R^d$ such that $D(x,y)=\norm{f(x)-f(y)}_p$ for every $x,y\in X$.

Our main theorem is an upper bound for general $\ell_p$ metrics. Let $\mathcal D_p(n,d,\Phi)$ denote the family of all $d$-dimensional $\ell_p$-metrics on $X$ with spread at most $\Phi$. (The dependence on $X$ is omitted in order to keep the notation simple. Recall it is an arbitrary set of $n$ labels.)
\begin{theorem}\label{thm:main}
For every $1\leq p\leq\infty$,
\ifprocs
\[ b_{1+\epsilon}(\mathcal D_p(n,d,\Phi)) = \]
\[ O(n(d+\log n)\log(1/\epsilon) + n\log\log \Phi) . \]
\else
\[ b_{1+\epsilon}(\mathcal D_p(n,d,\Phi)) = O(n(d+\log n)\log(1/\epsilon) + n\log\log \Phi) . \]
\fi
The summary and estimation algorithms are deterministic and run in time $\mathrm{poly}(n,d,\epsilon^{-1},\log\Phi)$.
\end{theorem}

or Euclidean metrics ($p=2$), we can first apply the Johnson-Lindenstrauss theorem on the pointset (distorting the distances by at most $(1+\epsilon)$) and then apply \cref{thm:main}.

\begin{theorem}[Euclidean metrics]\label{cor:euclidean}
\ifprocs
\[ b_{1+\epsilon}(\mathcal D_2(n,d,\Phi)) = \]
\[ O(\epsilon^{-2}\log(1/\epsilon) \cdot n\log n + n\log\log \Phi) . \]
\else
\[ b_{1+\epsilon}(\mathcal D_2(n,d,\Phi)) = O(\epsilon^{-2}\log(1/\epsilon) \cdot n\log n + n\log\log \Phi) . \]
\fi
\end{theorem}

Note that since the Johnson-Lindenstrauss theorem is randomized, then so is the resulting summary algorithm for Euclidean
metrics. This means that with probility $1/\mathrm{poly}(n)$, it may output a sketch that distorts the distances by more than a $(1+\epsilon)$ factor. However, this does not effect the sketch size nor the running time.

We complement this upper bound by a matching lower bound up to the $\log(1/\epsilon)$ term, see \cref{thm:euclidean_lowerbound}.

Another notable special case is $\ell_1$-metrics. By known embedding results, both our upper and lower bounds on $\mathcal D_2(n,d,\Phi)$ hold for $\mathcal D_1(n,d,\Phi)$ as well, see \Cref{sec:l1_metrics} for details.

For general metrics, we obtain tight upper and lower bounds. Let $\mathcal D_{\mathrm{all}}(n,\Phi)$ be the family of all metrics on $X$ with spread at most $\Phi$. Since any metric on $n$ points is an $n$-dimensional $\ell_\infty$-metric, we obtain the following corollary from \cref{thm:main}.
\begin{theorem}[general metrics]\label{cor:general}
\[ b_{1+\epsilon}(\mathcal D_{\mathrm{all}}(n,\Phi)) = O(n^2\log(1/\epsilon) + n\log\log \Phi) . \]
\end{theorem}
We show a tight lower bound in \cref{thm:general_lowerbound}.

\section{Summary Algorithm}
In this section we begin the proof of~\cref{thm:main}, by describing and analyzing the summary algorithm. In~\Cref{sec:estimation} we describe and analyze the estimation algorithm, and in~\Cref{sec:runningtimes} we discuss the running times of both algorithms.

Let $D$ be a $d$-dimensional $\ell_p$-metric on $X$, and let $f:X\rightarrow\R^d$ denote the available embedding, meaning $D(x,y)=\norm{f(x)-f(y)}_p$ for all $x,y\in X$. Throughout the proof we write $\norm{\cdot}$ for the $\ell_p$-norm $\norm{\cdot}_p$, omitting the subscript. We assume by normalization that $\min_{x,y\in X,x\neq y}\norm{f(x)-f(y)}=1$, and thus $\Phi$ is an upper bound on the diameter.

We use the following variant of hierarchically well-separated trees (HSTs).~\cite{bartal1996probabilistic}
\begin{definition}\label{def:hst}
Let $T$ be a rooted, edge-weighted tree. Number the levels in $T$ bottom-up, starting with the deepest leaf which is defined to belong to the level $0$. Denote by $\ell(v)$ the level of every node $v$ in $T$.

We say $T$ is a \emph{$k$-hierarchically well-separated tree ($k$-HST)} if for every node $v$ in $T$, each edge connecting $v$ to a child has weight $k^{\ell(v)}$.
\end{definition}

\subsection{Building the Tree}

We now construct a $2$-HST $T$ from $X$ in a bottom-up manner. For $i=0,1,\ldots,\log\mathrm{diam}(X)+1$,
\begin{itemize}
  \item Let $G_i(X,E_i)$ be the (unweighted) graph in which $x,y\in X$ are neighbors if $\norm{f(x)-f(y)}<2^i$. (Note that $E_0=\emptyset$, by our assumption that the minimum pairwise distance is $1$.)
  \item For every connected component $C$ in $G_i$ add a tree node $v$, and let $C(v):=C$.
  \item The connected components of $G_i$ form a partition of $X$, and if $i>0$, the partition at level $i-1$ is a refinement of the partition at level $i$. Add the corresponding tree edges. This means that for all tree nodes $v,u$ at levels $i$ and $i-1$ respectively, such that $C(u)\subset C(v)$, we attach $v$ as the parent of $u$.
\end{itemize}

\paragraph{Notation and terminology.} For every $v\in T$, we denote  its level in $T$ by $\ell(v)$. The \emph{degree} of $v$ is its number of children. The set $C(v)$ is the \emph{cluster} associated with $v$. We denote its diameter by
\[ \Delta(v) := \mathrm{diam}(C(v)). \]
Observe that the leaves in $T$ correspond bijectively to points in $X$, in the sense that for every $x\in X$ there is a unique leaf whose associated cluster is $\{x\}$. We denote  that leaf by $\mathrm{leaf}(x)$.

\begin{observation}\label{obs:separation}
If $x,y\in X$ are at different components of the partition induced by the level-$i$ nodes of $T$, then $\norm{f(x)-f(y)}\geq2^i$.
\end{observation}

\subsection{Compressing the Tree}
As constructed above, $T$ has $n$ leaves and up to $\log \Phi +2$ levels, and since it may contain degree-$1$ nodes its total size can be as large as $O(n\log \Phi)$. We wish to make it smaller by compressing long paths of degree-$1$ nodes.

A \emph{maximal $1$-path} in $T$ is a downward path $v_0,v_1,\ldots,v_k$ such that $v_1,\ldots,v_{k-1}$ are degree-$1$ nodes, and $v_0$ and $v_k$ are not degree-$1$ nodes ($v_k$ may have degree $0$). For every such path in $T$, if $k>\log(\frac{\Delta(v_k)}{2^{\ell(v_k)}})+\log(1/\epsilon)$, we replace the path from $v_1$ to $v_k$ with a \emph{long edge} directly connecting $v_1$ to $v_k$. (Note that the edge $v_1$ remains a degree-$1$ node in the tree.) We mark it as long and store with it the original path length, $k$. Non-long edges will be called \emph{short edges}.

\begin{lemma}\label{clm:compressed_tree}
The tree after compression has at most $2n(2+\log(1/\epsilon))$ nodes.
\end{lemma}
\begin{proof}
We charge the degree-$1$ nodes on every $1$-path to the bottom node of the path. The total number of nodes in the tree can then be written as $\sum_{v:\mathrm{deg}(v)\neq1}k(v)$, where $k(v)$ is the length of the maximal $1$-path whose bottom node is $v$. Due to the compression we have $k(v)\leq\log(\frac{\Delta(v_k)}{2^{\ell(v_k)}})+\log(1/\epsilon)$. Since the tree has $n$ leaves, it has at most $2n$ non-degree-$1$ nodes, so the total contribution of the second term is at most $2n\log(1/\epsilon)$. For the total contribution of the first term, we need to show
\begin{equation}\label{eq:longedges}
\sum_{v:\mathrm{deg}(v)\neq1}\log\left(\frac{\Delta(v)}{2^{\ell(v)}}\right) \leq 4n.
\end{equation}
To this end, consider the original tree $T$ (before compression) and contract every edge whose top (parent) node has degree $1$. Do this repeatedly, until there are no more degree-$1$ nodes, to obtain a tree $T'$. When contracting an edge $u\rightarrow v$, we identify the contracted node with the bottom node $v$ of the original edge, and it keeps its $\Delta(v)$ and $\ell(v)$ that were set in $T$ (note that $\ell(v)$ continues to denote the level of $v$ in $T$, and not in $T'$). It is clearly sufficient to prove \cref{eq:longedges} for $T'$ instead of $T$, since we have only removed degree-$1$ nodes in the transition. Also note that $T'$ has $n$ leaves and no degree-$1$ nodes, and hence at most $2n$ nodes in total.

For every node $v\in T'$, $\Delta(v)$ is upper-bounded by the sum of the edge weights in the subtree of $T'$ rooted at $v$, where the weight of an edge $u\rightarrow u'$ in $T'$ is $2^{\ell(u)}$ (see~\cref{def:hst} and recall that $\ell(u)$ denotes the level of $u$ in $T$). To see this, consider an edge $u\rightarrow v$. By construction of $T$, this means the cluster $C(v)$ has been merged into the larger cluster $C(u)$, when the distance from $C(v)$ to $C(u)\setminus C(v)$ was at most $2^{\ell(u)}$. Hence the edge weight, which is $2^{\ell(u)}$, bounds the contribution of that merging to the diameter of $C(u)$. However, if $u$ is a degree-$1$ node, then $C(u)=C(v)$ and no merging has been performed, so there is no contribution to the diameter of $C(u)$ that needs to be accounted for. In sum, only those edges whose top nodes has degree different than $1$ are needed in order to bound the cluster diameters, and these are exactly the edges in $T'$. See \Cref{fig:compression} for illustration.

Consequently, it is sufficient to prove
\[ \sum_{v\in T'}\log\left(\frac{\mathrm{wt}(v)}{2^{\ell(v)}}\right) \leq 4n, \]
where the \emph{weight} $\mathrm{wt}(v)$ of a node $v$ in $T'$ is the sum of the edge weights in its subtree. We will prove the stronger bound,
\begin{equation}\label{eq:longedges_aux}
\sum_{v\in T'}\frac{\mathrm{wt}(v)}{2^{\ell(v)}} \leq 4n,
\end{equation}
by summing over edges. An edge in $T'$ contributes $1$ to the term $\frac{\mathrm{wt}(v)}{2^{\ell(v)}}$ of its parent $v$ (recall that the edge weight is $2^{\ell(v)}$), $1/2$ to the term $\frac{\mathrm{wt}(v')}{2^{\ell(v')}}$ of its grandparent $v'$, $1/4$ to the great-grandparent term, and so on. In total, each edge contributes at most $2$ to the sum in \cref{eq:longedges_aux}, and since $T'$ has at most $2n$ edges, the sum is bounded by $4n$.
\ifprocs
\qedsymbol
\fi
\end{proof}

\begin{figure}
    \centering
    \begin{subfigure}[b]{0.3\textwidth}
        \includegraphics[width=0.7\textwidth]{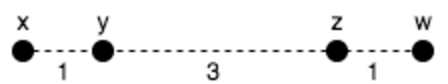}
        \caption{The metric space}
        \label{fig:metric_space}
    \end{subfigure}
 	\quad
    \begin{subfigure}[b]{0.3\textwidth}
        \includegraphics[width=0.7\textwidth]{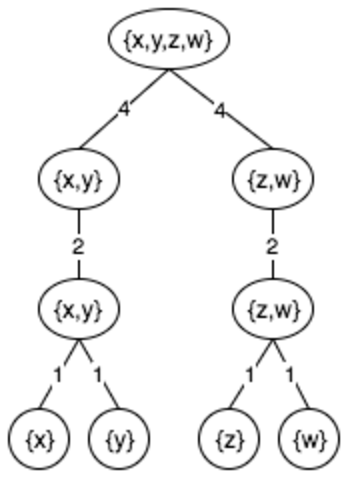}
        \caption{The tree $T$}
        \label{fig:hst}
    \end{subfigure}
    \quad
    \begin{subfigure}[b]{0.3\textwidth}
        \includegraphics[width=0.7\textwidth]{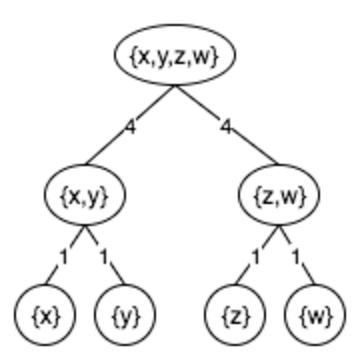}
        \caption{The contracted tree $T'$}
        \label{fig:contracted_hst}
    \end{subfigure}
    \caption{The metric space consists of $4$ colinear points at distances as indicated in (a). In the contracted tree $T'$, the diameter of every cluster is bounded by the sum of edge weights (written on the edges) in the corresponding subtree. }
    \label{fig:compression}
\end{figure}

From now on $T$ will denote the tree after compression. We will often partition it into \emph{subtrees} by removing the long edges.

\subsection{Centers}
With every node $v$ in $T$ we now associate a \emph{center} $c(v) \in X$, which will be a representative point for the associated cluster $C(v)$. We choose the centers by the following bottom-up process on $T$:
\begin{itemize}
  \item If $v$ is a leaf in $T$, i.e.~$v=\mathrm{leaf}(x)$ for some $x\in X$, then set $c(v):=x$.
  \item If $v$ is the top node of a long edge, then recall it is the unique edge outgoing from $v$ (since by construction top nodes of long edges have degree $1$). Denote by $u$ the bottom node of the long edge, and set $c(v):=c(u)$.
  \item Otherwise, $v$ has children $v_1,\ldots,v_k$ connected to it by short edges. Recall that in the graph $G_{\ell(v)}$, $C(v)$ is a connected component that contains each of $C(v_1),\ldots,C(v_k)$. By contracting each of those clusters in $G_{\ell(v)}$ into a single node, we get a connected graph whose nodes correspond to $v_1,\ldots,v_k$. Fix an arbitrary rooted spanning tree of this graph, and denote it $\tau(v)$. Suppose w.l.o.g.~that the root is $v_1$. Set $c(v):=c(v_1)$.
\end{itemize}

For every node $v$ in $T$ we have just fixed a rooted tree $\tau(v)$ on its children $v_1,\ldots,v_k$. We will use those trees later in the construction. To make the text clearer, we will refer to the parent of $v_i$ in $\tau(v)$ as \emph{$\tau$-predecessor} of $v_i$ (which is another child of $v$ in $T$). In contrast, the term \emph{parent} of $v_i$ will be reserved for its parent in $T$ (which is $v$). See \Cref{fig:tautrees} for illustration.

\begin{figure}
    \centering
    \includegraphics[scale=0.6]{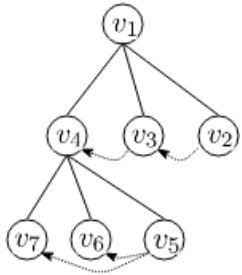}
    \caption{The solid arcs represent the tree $T$, and the dashed arrows represent the trees $\tau(v_1)$ and $\tau(v_4)$, defined on their children in $T$. The tree $\tau(v_1)$ is the path $v_2\rightarrow v_3\rightarrow v_4$. The tree $\tau(v_4)$ is the star with center $v_5$ pointing at $v_6$ and $v_7$. The \emph{parent} of $v_3$ is $v_1$, while its \emph{$\tau$-predecessor} is $v_2$.}
    \label{fig:tautrees}
\end{figure}

\subsection{Ingresses}\label{sec:ingresses}
With every node $v$ in $T$ we will now associate an \emph{ingress} node, $in(v)\in T$. Intuitively, the idea is to store the location of $c(v)$ as its displacement from $c(in(v))$, the center of the ingress node. Therefore we would like the ingress to be a tree node whose center is close to $c(v)$, and such that we have available an approximate location for it.

The ingresses are chosen separately within each subtree, where we recall that the \emph{subtrees} are formed from $T$ by removing the long edges. For the root of the subtree we do not set an ingress as we will not need one. Now suppose we have a node $v$ with children $v_1,\ldots,v_k$ in the same subtree (i.e.~connected to $v$ with short edges). Recall we have a tree $\tau(v)$ on $v_1,\ldots,v_k$, rooted at $v_1$, where an edge in $\tau$ connecting $v_i,v_j$ means that 
\begin{equation}\label{eq:clustertree}
2^{\ell(v)-1} \leq \mathrm{dist}(C(v_i),C(v_j)) < 2^{\ell(v)} ,
\end{equation}
where
\ifprocs
\[ \mathrm{dist}(C(v_i),C(v_j)) := \]
\[ \min\{ \norm{f(x)-f(y)} : x\in C(v_i), y\in C(v_j) \} . \]
\else
\[ \mathrm{dist}(C(v_i),C(v_j)) := \min\{ \norm{f(x)-f(y)} : x\in C(v_i), y\in C(v_j) \} . \]
\fi

For $v_1$, we set $in(v_1)=v$. For $v_i$ with $i>1$, let $v_j$ be the $\tau$-predecessor of $v_i$. Let $y_i\in C(v_j)$ be the closest point to $C(v_i)$ in $C(v_j)$. Note that in $T$, there is a downward path from $v_j$ to $\mathrm{leaf}(y_i)$. We set $in(v_i)$ to be the lowest point on that path that does not go through any long edge.

The motivation for this choice is that ideally we would like $\mathrm{leaf}(y_i)$ to be the ingress of $v_i$, but $\mathrm{leaf}(y_i)$ might be outside the current subtree. As we will see next, we will have approximate locations relative to $c(v_i)$ only for nodes in the same subtree as $v_i$. Therefore we choose the ingress as the node in the current subtree whose center is closest to $y_i$.

The following lemma gives us a bound on the distance between the node center and its ingress center.

\begin{lemma}\label{lmm:ingress}
For every $u\in T$ which is not a root of a subtree, we have
\begin{equation}\label{eq:ingress_goal}
\norm{f(c(u))-f(c(in(u)))} \leq 3\cdot2^{\ell(u)} + \Delta(u) .
\end{equation}
\end{lemma}
\begin{proof}
We use the same notation as in the above choice of ingresses. Suppose we have a node $v$ with children $v_1,\ldots,v_k$, and we wish to prove the bound for some $u=v_i$. For $v_1$ we have $in(v_1)=v$ and (by choice of centers) $c(v_1)=c(v)$, hence $c(v_1)=c(in(v_1))$ and \cref{eq:ingress_goal} holds for $u=v_1$ trivially.

Now suppose $i>1$. By \cref{eq:clustertree}, the point $y_i$ satisfies
\[ \norm{f(c(v_i))-f(y_i)} \leq 2^{\ell(v)} + \Delta(v_i). \]
Noting that $\ell(v_i)=\ell(v)-1$, we have \begin{equation}\label{eq:ingress_aux}
\norm{f(c(v_i))-f(y_i)}\leq 2\cdot2^{\ell(v_i)}+\Delta(v_i).
\end{equation}
We have set $in(v_i)$ to be the lowest node in the path from $v_j$ to $\mathrm{leaf}(y_i)$ that does not traverse a long edge. We consider two cases:
\begin{itemize}
  \item The path has no long edges, which means $in(v_i)=\mathrm{leaf}(y_i)$. Then $c(in(v_i))=y_i$, and \cref{eq:ingress_goal} for $u=v_i$ follows from \cref{eq:ingress_aux}.
  \item The path has long edges, which means $in(v_i)$ is the top node of a long edge. Let $k$ be its original length and $w$ its bottom node. Note that $c(in(v_i))=c(w)$. Then
  \[
    \norm{f(y_i) - f(c(in(v_i)))} =
    \norm{f(y_i) - f(c(w))} \leq
    \Delta(w)
\]
\[
    = 2^{\ell(w)+\log(\frac{\Delta(w)}{2^{\ell(w)}})} <
    2^{\ell(w)+k} = 
    2^{\ell(in(v_i))} \leq
    2^{\ell(v_i)}.
  \]
  Combining this with \cref{eq:ingress_aux} yields \cref{eq:ingress_goal} for $u=v_i$.
  \ifprocs
  \qedsymbol
  \fi
\end{itemize}
\end{proof}

We also state the following fact.
\begin{claim}\label{clm:ingress_level}
For a node with an ingress, $\ell(in(v)) \leq \ell(v)+1$.
\end{claim}
\begin{proof}
By construction, $in(v)$ is either the parent of $v$ in $T$, or a descendant of the parent.
\ifprocs
\qedsymbol
\fi
\end{proof}

\subsection{Surrogates}\label{sec:surrogates}
We now associate a \emph{surrogate} $s^*(v)\in\R^d$ with each tree node $v$, which will be an approximate location for its center $c(v)$. The goal is to choose the surrogates such that the distances between them can be recovered from the sketch, thus approximating the distances between the actual points in $X$.

For $\delta>0$ and $B\subset\R^d$, recall that $N\subset\R^d$ is a \emph{$\delta$-net} for $B$ if for every $q\in B$ there is $\bar{q}\in N$ such that $\norm{q-\bar{q}} \leq \delta$. We use the following known result.

\begin{lemma}\label{lmm:epsnet}
For every $\delta>0$ there is a $\delta$-net $\mathcal N_\delta$ for the unit ball in $\ell_p^d$, of size $O(1/\delta)^d$. 
\end{lemma}

Let us first give an intuitive description of the choice of surrogates. Take a node $v$ and put $q:=f(c(v))$ for brevity; this is the node location in $\R^d$. We wish to approximately store $q$ with a small number of bits. To this end we pick a point $\bar q$ close to $q$, i.e.~such that $\sigma:=q-\bar q$ has small norm. We then round $\sigma$ to a vector $\tilde\sigma$ using a $\delta$-net, and use $\bar q+\tilde\sigma$ as the surrogate.

The natural choice for $\bar q$ is the ingress of $v$. \Cref{lmm:ingress} then gives a bound on $\norm{\sigma}$, which lets us pick $\delta$ that provides satisfactory approximation while keeping the storage cost of $\tilde\sigma$ small. However, in order to recover the surrogate we also need to store $\tilde q$, the location of the ingress, which is too costly. Instead, we choose $\bar q$ inductively as the surrogate of the ingress, $\bar q:=s^*(in(v))$.

We proceed to the formal construction. The surrogates are defined independently in each subtree. Within a subtree $T_{\mathrm{sub}}$ of $T$, we wish to define $s^*(v)$ inductively from $s^*(in(v))$, so we need an ordering for the induction such that a node is always processed after its ingress. We can achieve this by traversing $T_{\mathrm{sub}}$ in a DFS order, with the order of traversing the children of each node $v$ (with degree greater than $1$) being  top-down on $\tau(v)$. Put differently, when we traverse a node $v$ we first process it, and then (recursively) traverse its children in a top-down order by $\tau(v)$. This means that whenever we process a node $v$, both its parent $v'$ in $T_{\mathrm{sub}}$ and its $\tau$-predecessor $v_{\tau}$ have already been traversed. Since it is a DFS scan, and $v_{\tau}$ is a sibling of $v$ in $T_{\mathrm{sub}}$, this means all descendants of $v_{\tau}$ in $T_{\mathrm{sub}}$ have already been processed. In particular, since $in(v)$ is by construction either $v'$ or a descendant of $v_{\tau}$ in $T_{\mathrm{sub}}$, it means $in(v)$ has already been processed. As we will refer to this ordering again later on, we call it for brevity the \emph{$\tau$-DFS} ordering of the nodes in a subtree of $T$.

We now define the induction steps. Denote
\[ \delta(v) := \left(5+\lceil\frac{\Delta(v)}{2^{\ell(v)}}\rceil\right)^{-1}. \]

\paragraph{Induction base:} For the root $v$ of the subtree, set $s^*(v)=f(c(v))$.

\paragraph{Inductive step:} For a non-root $v$,
\begin{itemize}
  \item Let $\mathrm{disp}(v) := f(c(v)) - s^*(in(v))$ be the displacement from the ingress' surrogate.
  \item Let $\eta^*(v) := \frac{\delta(v)}{2^{\ell(v)}}\cdot\mathrm{disp}(v)$ be the normalized displacement. (We will soon show $\norm{\eta^*(v)}\leq1$.)
  \item Let $\eta(v)$ be the closest point to $\eta^*(v)$ in the net $\mathcal N_{\delta(v)}$.
  \item Finally, the surrogate is $s^*(v) := s^*(in(v)) + \frac{2^{\ell(v)}}{\delta(v)}\cdot\eta(v)$.
\end{itemize}

\begin{lemma}\label{lmm:surrogate}
For every $v\in T$, $\norm{f(c(v))-s^*(v)} \leq 2^{\ell(v)}$.
\end{lemma}
\begin{proof}
By induction on the $\tau$-DFS ordering within each subtree. In the base case, $v$ is the root and then the claim is trivial since $s^*(v)=f(c(v))$. Now suppose $v$ is not the root. By induction on the ingress we have
\[ \norm{f(c(in(v)))-s^*(in(v))} \leq 2^{\ell(in(v))} , \]
and then by \cref{clm:ingress_level},
\[ \norm{f(c(in(v)))-s^*(in(v))} \leq 2 \cdot 2^{\ell(v)} . \]
By \cref{lmm:ingress},
\[ \norm{f(c(v))-f(c(in(v)))} \leq 3\cdot2^{\ell(v)} + \Delta(v), \]
and together,
\[ \norm{f(c(v))-s^*(in(v))} \leq 5\cdot2^{\ell(v)} + \Delta(v) \leq \frac{2^{\ell(v)}}{\delta(v)}. \]
This implies $\norm{\eta^*(v)}\leq1$, and since $\mathcal N_{\delta(v)}$ is a net for the unit ball, this ensures $\norm{\eta^*(v)-\eta(v)}\leq\delta(v)$. Finally,
\ifprocs
\begin{align*}
  & \norm{f(c(v))-s^*(v)} \\
  &= \norm{f(c(v)) - s^*(in(v)) - \tfrac{2^{\ell(v)}}{\delta(v)}\cdot\eta(v)} \\
  &= \norm{f(c(v)) - s^*(in(v)) - \\
  & \;\;\;\;\;\; \tfrac{2^{\ell(v)}}{\delta(v)}\cdot(\eta^*(v) \eta^*(v)+\eta(v))} \\
  &= \norm{\tfrac{2^{\ell(v)}}{\delta(v)}\cdot(\eta(v)-\eta^*(v))} \\
  &\leq \tfrac{2^{\ell(v)}}{\delta(v)}\cdot\delta(v) \\
  &= 2^{\ell(v)}. \text{\;\;\;\;\;\;\qedsymbol}
\end{align*}
\else
\begin{align*}
  \norm{f(c(v))-s^*(v)} &= \norm{f(c(v)) - s^*(in(v)) - \tfrac{2^{\ell(v)}}{\delta(v)}\cdot\eta(v)} \\
  &= \norm{f(c(v)) - s^*(in(v)) - \tfrac{2^{\ell(v)}}{\delta(v)}\cdot(\eta^*(v)-\eta^*(v)+\eta(v))} \\
  &= \norm{\tfrac{2^{\ell(v)}}{\delta(v)}\cdot(\eta(v)-\eta^*(v))} \\
  &\leq \tfrac{2^{\ell(v)}}{\delta(v)}\cdot\delta(v) \\
  &= 2^{\ell(v)}.
\end{align*}
\fi
\end{proof}

For the leaves of each subtree we will actually use a better $\delta(v)$,
\[ \delta'(v) := \delta(v)\cdot\epsilon. \]
Then the previous lemma yields
\begin{corollary}\label{cor:surrogate_leaf}
For $v\in T$ which is a leaf in its subtree, $\norm{f(c(v))-s^*(v)} \leq 2^{\ell(v)}\cdot\epsilon$.
\end{corollary}
The corollary follows by simply executing the last round of induction in the proof of \cref{lmm:surrogate} with the improved $\delta(v)$.

\subsection{The Sketch}\label{sec:thesketch}
In the sketch we store the following information:
\begin{itemize}
  \item The tree $T$. For each edge we store whether it is short or long, and for the long edges we store their original lengths.
  \item For every tree node $v$ we store the center label $c(v)$, the ingress label $in(v)$, the value $ֿ\delta(v)^{-1}$ (which is the integer $5+\lceil\frac{\Delta(v)}{2^{\ell(v)}}\rceil$), and the approximate displacement $\eta(v)$, encoded as an element of $N_{\delta(v)}$ (or $N_{\delta'(v)}$, if $v$ is a leaf in its subtree).
\end{itemize}

The purpose of storing the lengths of long edges is to compute the levels $\ell(v)$, which we recall are the levels in the uncompressed tree. They are needed in order to recover the surrogates (up to a shift), as will be discussed in \Cref{sec:estimation}.

We now bound the total size of the sketch. We start with the following observation.
\begin{claim}\label{clm:subtree_leaves}
$(i)$ The number of long edges in $T$ is at most $2n$.

$(ii)$ The number of nodes in $T$ which are leaves in their subtree is at most $3n$.
\end{claim}
\begin{proof}
For part $(i)$, recall that the bottom node of every long edge had degree different than $1$ in the original tree (before compression). Since that tree had $n$ leaves, it could only have $2n$ such nodes. Part $(ii)$ follows from $(i)$ by noting that each node in $T$ which is a leaf in its subtree is either a leaf in the original (non-compressed) tree, or the top node of a long edge.
\ifprocs
\qedsymbol
\fi
\end{proof}

\begin{lemma}\label{lmm:sketch_size}
The total sketch size is $O(n(d+\log n)\log(1/\epsilon) + n\log\log \Phi)$.
\end{lemma}
\begin{proof}
We start by analyzing the space needed to store the tree structure. By \cref{clm:compressed_tree} the compressed tree $T$ has size $O(n\log(1/\epsilon))$, so its structure can be stored using $O(n\log(1/\epsilon))$ bits. The length of each long edge is bounded by the height of the original tree, which by construction is at most $\log \Phi+1$, so by \cref{clm:subtree_leaves} the total storage cost of the lengths is at most $2n\log(\log\Phi+1)$ bits. Overall, the tree structure requires $O(n\log(1/\epsilon)+n\log\log \Phi)$ bits to store. We now analyze the cost of the information stored for each node.
\begin{itemize}
  \item Centers: Each center is a label in $X$ and hence takes $\log n$ bits to store. In total, $O(n\log(1/\epsilon)\cdot\log n)$ bits.
  \item Ingresses: $in(v)$ is a node in $T$, but we can further observe that $in(v)$ is either the parent of $v$ or a node which is a leaf in its subtree. Therefore by \cref{clm:subtree_leaves} the ingress is one of $O(n)$ possible nodes, and takes $O(\log n)$ bits to store. In total, $O(n\log(1/\epsilon)\cdot\log n)$ bits.
  \item Precisions: Their total storage cost is
\ifprocs
  \[
    \sum_{v\in T}\log\left(\frac{1}{\delta(v)}\right) = 
    \sum_{v\in T}\log\left(5+\lceil\frac{\Delta(v)}{2^{\ell(v)}}\rceil\right)
  \]
  \[
    \leq 3|T| + \sum_{v\in T}\log\left(\frac{\Delta(v)}{2^{\ell(v)}}\right).
  \]
\else
  \[
    \sum_{v\in T}\log\left(\frac{1}{\delta(v)}\right) = 
    \sum_{v\in T}\log\left(5+\lceil\frac{\Delta(v)}{2^{\ell(v)}}\rceil\right) \leq
    3|T| + \sum_{v\in T}\log\left(\frac{\Delta(v)}{2^{\ell(v)}}\right).
  \]
\fi
  Since $\sum_{v\in T}\log\left(\frac{\Delta(v)}{2^{\ell(v)}}\right) = O(|T|) = O(n\log(1/\epsilon))$ (see the proof of \cref{clm:compressed_tree}), the total storage cost of the precisions is $O(n\log(1/\epsilon))$ bits.
  \item Displacements: By \cref{lmm:epsnet}, $\eta(v)$ is a point in a set of size $O(1/\delta(v))^d$, hence storing $\eta(v)$ takes $d\log(1/\delta(v))$ bits. Summing over all $v\in T$ we get $O(dn\log(1/\epsilon))$ bits, as shown above for the precisions. For the leaves of every subtree we kept a displacement up to an improved approximation, $\delta'(v)=\delta(v)\cdot\epsilon$. This adds $d\log(1/\epsilon)$ bits per $v$, and since by \cref{clm:subtree_leaves} there are $O(n)$ such nodes, in total this consumes additional $O(nd\log(1/\epsilon))$ bits.
\end{itemize}
In total, $O(n(d+\log n)\log(1/\epsilon) + n\log\log \Phi)$ bits.
\ifprocs
\qedsymbol
\fi
\end{proof}

\section{Estimation Algorithm}\label{sec:estimation}
We now show how to use the sketch to produce a $(1\pm\epsilon)$-approximation for the distance between any two points in $X$. The key point is that within each subtree, we can recover the surrogates up to a fixed (unknown) shift from the sketch. Formally, for every $v\in T$ we define the \emph{shifted surrogate} $s(v)\in\R^d$:
\begin{itemize}
  \item If $v$ is the root of its subtree, set $s(v):=\mathbf0$ (the origin in $\R^d$).
  \item Otherwise, set $s(v):=s(in(v))+\frac{2^{\ell(v)}}{\delta(v)}\cdot\eta(v)$.
\end{itemize}
Observe that we can indeed compute the shifted surrogate from the sketch: For every $v$ we have stored explicitly $in(v)$, $\ell(v)$ (inherent in storing the tree structure), $\delta(v)^{-1}$, and an encoding of $\eta(v)$ as an element in $\mathcal N_{\delta(v)}$ that can now be decoded. With those at hand, we can compute the shifted surrogates inductively in the $\tau$-DFS order on the subtree.

Furthermore, by comparing this construction to that of \Cref{sec:surrogates}, it is straightforward to see that for every node $v\in T$ we have $s(v)=s^*(v)-s^*(r)$, where $r$ is the root of the subtree in which $v$ resides. This means that the shifted surrogates within every subtree are the same as the original surrogates up to a fixed shift $s^*(r)$ (which cannot be recovered from the sketch, since it equals $f(c(r))$ and we never stored the true embedding of any point in the sketch). Hence,
\begin{claim}\label{clm:shifted_surrogate}
For every $v,v'\in T$ which are in the same subtree, $\norm{s(v)-s(v')}=\norm{s^*(v)-s^*(v')}$.
\end{claim}

Now given $x,y\in X$, we show to how to compute from the sketch a $(1\pm\epsilon)$-estimate for $\norm{f(x)-f(y)}$. Let $u$ be the lowest common ancestor of $\mathrm{leaf}(x)$ and $\mathrm{leaf}(y)$. Let $v_x$ be the lowest node on the path from $u$ down to $\mathrm{leaf}(x)$ that does not traverse a long edge. Similarly define $v_y$ for $y$. Note that $u,v_x,v_y$ are all in the same subtree, and $v_x,v_y$ are leaves in that subtree. See \Cref{fig:estimation} for illustration. The estimate we return is $\norm{s(v_x)-s(v_y)}$. By \cref{clm:shifted_surrogate} it equals $\norm{s^*(v_x)-s^*(v_y)}$, so our goal is to prove
\begin{equation}\label{eq:goal}
\norm{s^*(v_x)-s^*(v_y)} = (1\pm O(\epsilon)) \cdot \norm{f(x)-f(y)} .
\end{equation}
By the triangle inequality we have
\ifprocs
\begin{equation}\label{eq:estimate_aux}
  \norm{s^*(v_x)-s^*(v_y)} =
\end{equation}
\[   \norm{f(x)-f(y)} \pm \left( \norm{f(x)-s^*(v_x)} + \norm{f(y)-s^*(v_y)} \right) . \]
\else
\begin{equation}\label{eq:estimate_aux}
  \norm{s^*(v_x)-s^*(v_y)} = 
  \norm{f(x)-f(y)} \pm \left( \norm{f(x)-s^*(v_x)} + \norm{f(y)-s^*(v_y)} \right) .
\end{equation}
\fi
Now consider two cases for $v_x$:
\begin{itemize}
  \item If $v_x=\mathrm{leaf}(x)$ then $c(v_x)=x$, and hence by \cref{cor:surrogate_leaf}, $\norm{f(x)-s^*(v_x)}\leq2^{\ell(v_x)}\epsilon$.
  \item Otherwise, $v_x$ is the top node of a long edge. Let $k$ be its original length and $w_x$ its bottom node. Recall that by the construction, $k > \log(\Delta(w_x)/2^{\ell(w_x)})+\log(1/\epsilon)$. Also note that $c(v_x)=c(w_x)$. Then
\ifprocs
  \[
    \norm{f(x) - f(c(v_x))} = 
    \norm{f(x) - f(c(w_x))} \leq
    \Delta(w_x) =
  \]
  \[
    2^{\ell(w_x)+\log(\frac{\Delta(w_x)}{2^{\ell(w_x)}})} <
    2^{\ell(w_x)+k-\log(1/\epsilon)} = 
    2^{\ell(v_x)}\epsilon.
  \]
\else
  \[
    \norm{f(x) - f(c(v_x))} = 
    \norm{f(x) - f(c(w_x))} \leq
    \Delta(w_x) =
    2^{\ell(w_x)+\log(\frac{\Delta(w_x)}{2^{\ell(w_x)}})} <
    2^{\ell(w_x)+k-\log(1/\epsilon)} = 
    2^{\ell(v_x)}\epsilon.
  \]
\fi
\end{itemize}
Combining this with~\cref{cor:surrogate_leaf}, we get by the triangle inequality that $\norm{f(x)-s^*(v_x)}\leq2\cdot2^{\ell(v_x)}\epsilon$. This bound holds in both the the above cases. Similarly one shows $\norm{f(y)-s^*(v_y)}\leq2\cdot2^{\ell(v_y)}\epsilon$. Since $\ell(v_x)\leq\ell(u)-1$ and $\ell(v_y)\leq\ell(u)-1$, we can add these and obtain
\[ \norm{f(x)-s^*(v_x)} + \norm{f(y)-s^*(v_y)} \leq 2\cdot2^{\ell(u)}\epsilon. \]
By the construction of $T$, the fact that $u$ is the lowest common ancestor of $\mathrm{leaf}(x)$ and $\mathrm{leaf}(y)$ implies $\norm{f(x)-f(y)}\geq2^{\ell(u)-1}$ (see \cref{obs:separation}). Plugging this into the equation above yields
\[ \norm{f(x)-s^*(v_x)} + \norm{f(y)-s^*(v_y)} \leq \norm{f(x)-f(y)}\cdot4\epsilon, \]
and plugging this into \cref{eq:estimate_aux} proves \cref{eq:goal}, which proves \cref{thm:main}.

\begin{figure}
    \centering
    \includegraphics[scale=0.6]{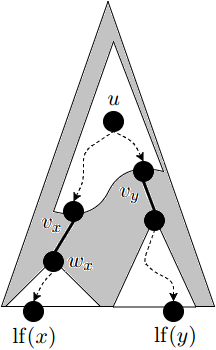}
    \caption{The estimate for $\norm{f(x)-f(y)}$ is $\norm{s(v_x)-s(v_y)}$. The external shaded triangle is the tree $T$. The white regions are subtrees. The dashed arrows are downward paths in $T$. The thick arcs are long edges.}
    \label{fig:estimation}
\end{figure}

\section{Running Times}\label{sec:runningtimes}
To analyze running times, we need an efficient version of \cref{lmm:epsnet}. We prove the following lemma in \Cref{sec:gridnet}.
\begin{lemma}\label{lmm:gridnet}
For every $\delta>0$, the $\ell_p$ unit ball in $\R^d$ has a $\delta$-net $\mathcal{N}_\delta$ such that
\begin{enumerate}
  \item Given $\eta^*\in\R^d$ with $\norm{\eta^*}_p\leq1$, one can find a $\delta$-close vector $\eta\in\mathcal{N}_\delta$ in time $O(d)$.
  \item A vector $\eta\in\mathcal{N}_\delta$ can be encoded as a bitstring of length $O(d\log(1/\delta))$, in time $O(\frac{1}{\delta}d^{1+1/p})$.
  \item Given the bitstring encoding as above, the coordinates of $\eta$ in $\R^d$ can be recovered in time $O(\frac{1}{\delta}d^{1+1/p})$.
\end{enumerate}
\end{lemma}

\paragraph{Summary time.}
We spend $O(n^2\log\Phi)$ time setting up the distances graph and building and compressing the tree. Then, the processing time for every node $v\in T$ is dominated by encoding the $\delta(v)$-net vectors, which by \cref{lmm:gridnet} takes time $O(d^{1+1/p}/\delta(v))$. Summing over the nodes, and recalling that $\delta(v) \geq \epsilon\cdot\left(5+\lceil\frac{\Delta(v)}{2^{\ell(v)}}\rceil\right)^{-1} \geq \epsilon\cdot\left(6+\frac{\Delta(v)}{2^{\ell(v)}}\right)^{-1}$, we get
\ifprocs
\begin{equation}\label{eq:netencoding}
  \sum_{v\in T}\frac{d^{1+1/p}}{\delta(v)} \leq \frac{d^{1+1/p}}{\epsilon}\left(6|T| + \sum_{v\in T}\frac{\Delta(v)}{2^{\ell(v)}}\right)
\end{equation}
\[  = \frac{d^{1+1/p}}{\epsilon}\cdot O\left(n\log\left(\frac{1}{\epsilon}\right)\right). \]
\else
\begin{equation}\label{eq:netencoding}
  \sum_{v\in T}\frac{d^{1+1/p}}{\delta(v)} \leq
  \frac{d^{1+1/p}}{\epsilon}\left(6|T| + \sum_{v\in T}\frac{\Delta(v)}{2^{\ell(v)}}\right) = 
  \frac{d^{1+1/p}}{\epsilon}\cdot O\left(n\log\left(\frac{1}{\epsilon}\right)\right).
\end{equation}
\fi
(See the proof of \cref{clm:compressed_tree} for the latter bound.) The total summary time is $O(n^2\log\Phi + nd^{1+1/p}\epsilon^{-1}\log(1/\epsilon))$.

\begin{observation}
Note that in the Euclidean case, the $n^2\log\Phi$ term in the running time bound can be reduced to $O(n^{1+\alpha}\log\Phi)$ for any constant $\alpha>0$, at the cost of increasing the sketch size by a multiplicative factor of $\alpha^{-1}$. (The Johnson-Lindenstrauss transform, which we also use as a preceding step, can be executed in time $O(\epsilon^{-2}n\log n)$ \cite{AilonC09}.) To this end, set $c:=\alpha^{-1/2}$. In constructing the tree, we use the algorithm of \cite{HarpeledIM12} to compute $c$-approximate connected components in each level. Their algorithm is based on Locality-Sensitive Hashing (LSH), which in Euclidean spaces can be implemented in time $O(n^{1+1/c^2})$ \cite{AndoniI06}. Using $c$-approximate connected components means that clusters in level $\ell$ of the tree can be merged if the distance between them is at most $c\cdot2^\ell$ (rather than just $2^\ell$), and to account for this constant loss, we need to scale $\epsilon$ down to $\epsilon/c$. Since the dependence of the sketch size on $\epsilon$ is $\log(1/\epsilon)/\epsilon^2$, the multiplicative loss in the sketch size is $c^2=\alpha^{-1}$.
\end{observation}

\paragraph{Estimation Time.}
Since the height of the tree is at most $\log\Phi+2$, we spend $O(\log\Phi)$ time finding the lowest common ancestor of $\mathrm{leaf}(x),\mathrm{leaf}(y)$ and finding $v_x,v_y$. Then we need to compute the shifted surrogates $s(v_x),s(v_y)$. Due to the inductive definition of the $s(v_x)$, in order to compute $s(v_x)$ we need to traverse $\tau$-predecessors backwards until we reach the root of the subtree, whose shifted surrogate is known to be $\mathbf0$. In the worst case we might traverse all nodes in $T$. For each node $v$ we need to decode the $\delta(v)$-net vector $\eta(v)$, which by \cref{lmm:gridnet} takes time $O(d^{1+1/p}/\delta(v))$. Applying \cref{eq:netencoding} again, we see that the total estimation time is $O(\log\Phi + nd^{1+1/p}\epsilon^{-1}\log(1/\epsilon))$.

In practical settings such query time is often considered prohibitive. We now describe a modification to our scheme that yields a different trade-off between the sketch size and the estimation time. In particular, letting
\[ K:= \lceil\log(2\cdot\Phi\cdot \epsilon^{-1}\cdot d^{1/p})\rceil , \]
we show how to achieve estimation time of $O(\log\Phi + dK)$ in the expense of increasing the sketch size by a factor of $\log d$. To demonstrate why this is beneficial, consider a typical Euclidean setting in which $d=O(\epsilon^{-2}\log n)$ (by Johnson-Lindenstrauss dimension reduction) and $\Phi=\mathrm{poly}(n)$. \Cref{cor:euclidean} gives a sketch size of $O(\epsilon^{-2}\log(1/\epsilon)\cdot n\log n)$ bits with $\tilde O(\epsilon^{-4}n)$ estimation time.\footnote{We use $\tilde O(f)$ to denote $O(f\cdot\mathrm{polylog}(f))$.} The modification increases the sketch size by a factor of $O(\log\log n + \log(1/\epsilon))$, and improves the estimation time to $\tilde O(\epsilon^{-2}\log^2n)$.

The first estimation bottleneck is decoding the net vectors, and we resolve this by replacing the $\delta$-net from \cref{lmm:gridnet} with the uniform grid $(\frac{\delta}{d^{1/p}}\Z)^d$. In contrast, \cref{lmm:gridnet} uses the intersection of this grid with the unit ball. We can store a point in this grid using $O(d\log(1/\delta)+d\log d)$ bits without any encoding, which adds $O(d\log d)$ bits per point over \cref{lmm:gridnet}. In total, the sketch size increases by a factor of $O(\log d)$, and the processing time of a node $v$ decreases to $O(d)$.

The second estimation bottleneck is computing the shifted surrogates by induction on the $\tau$-predecessors all the way back to the subtree root. We resolve this by storing some shifted surrogates explicitly in the sketch. This is done separately in each subtree $T'$, as follows.
\begin{enumerate}
  \item Construct the tree $T_\tau'$ on the nodes of $T'$, by attaching each node as a child of its $\tau$-predecessor.
  \item Pick $\lceil |T_\tau'|/K \rceil$ nodes in $T_\tau'$, called \emph{landmark nodes}, such that for every $v\in T_\tau'$, we can reach a landmark node from $v$ by going upward in $T_\tau'$ at most $K$ steps. This can be done as follows: Start with a lowest node $v\in T_\tau'$; climb upward $K$ steps (or less if the root is reached), to a node $\hat{v}$; declare $\hat{v}$ a landmark node, remove it from $T_\tau'$ with all its decendants, and iterate. Since every iteration but the last removes at least $K$ nodes from $T_\tau'$, we finish with at most $\lceil |T_\tau'|/K \rceil$ landmark nodes.
\end{enumerate}
For every landmark node $\hat{v}$ we explicitly store in the sketch the shifted surrogate $s(\hat{v})$. Now, in order to compute $s(v)$ of any given node $v$, we need to trace the $\tau$-predecessors backward at most $K$ times until we reach a landmark nodes whose shifted surrogate is known. The computation time per node is $O(d)$, so in total, the resulting estimation time is $O(\log\Phi + dK)$.

It remains to verify that storing the shifted surrogates for the landmark nodes does not asymptotically increase the sketch size. To this end, fix a landmark node $\hat{v}$. Recall that the shifted surrogates are defined recursively, starting at $\mathbf0$ for the subtree root, and then in each step adding a vector of the form $\delta(v)^{-1}2^{\ell(v)}\eta(v)$ (see \Cref{sec:estimation}). Since $\eta(v)$ is a point on a grid with side either $\delta'(v)/d^{1/p}=\delta(v)\epsilon/d^{1/p}$ (if $v$ is a leaf in its subtree) or $\delta(v)/d^{1/p}$ (otherwise), we see that each step adds an integer multiple of either $\epsilon/d^{1/p}$ or $1/d^{1/p}$ to each coordinate of $s(\hat{v})$. On the other hand, since $\mathbf0$ is also a shifted surrogate (of the center of the root of the subtree in which $\hat{v}$ is present), we must have $\norm{s(\hat{v})-\mathbf0}\leq(1+\epsilon)\Phi$, and in particular each coordinate of $s(\hat{v})$ is bounded by $2\Phi$. Together, we see that each coordinate of $s(\hat{v})$ can be represented with $\lceil\log(2\Phi\cdot \epsilon^{-1}\cdot d^{1/p})\rceil=K$ bits. Multiplying by $d$ coordinates, we find that $O(dK)$ bits suffice to fully store any shifted surrogate. Since we are storing them for $O(|T|/K)$ landmark nodes, we spend additional $O(d|T|)=O(nd\log(1/\epsilon))$ bits, which does not asymptotically change the sketch size.

\section{Lower Bounds}

\begin{theorem}[Euclidean metrics]\label{thm:euclidean_lowerbound}
Fix $\gamma>0$. If $\epsilon\geq1/n^{0.5-\gamma}$, then
\[ b_{1+\epsilon}(\mathcal{D}_2(n,d,\Phi))=\Omega(\gamma\cdot\epsilon^{-2}n\log n+n\log\log \Phi) . \]
\end{theorem}
\begin{proof}
Denote $k:=1/\epsilon^2$. Note that since $\epsilon>1/\sqrt{n}$, we may assume w.l.o.g.~that $k$ is an integer. Let $B$ be the set of standard basis vectors in $\R^n$, and let $a_1,\ldots,a_n$ be an arbitrary sequence of $k$-sparse vectors in $\B^n$ (note that we allow repetitions). Denote $A:=\{a_1,\ldots,a_n\}$. We sketch the Euclidean metric on the set $(\frac{1}{\sqrt k}A)\cup B$ up to distortion $1\pm\frac{1}{2}\epsilon$. We can also keep track of repeating elements in $a_1,\ldots,a_n$ using $n\log n$ bits (details omitted).

For every $a_j$ and $i\in\{1,\ldots,n\}$ we have $a_j^Te_i=a_j(i)$ and hence
\[ \norm{\frac{1}{\sqrt k}a_j-e_i}_2^2 = 2-\frac{2}{\sqrt k}a_j(i) = 2-2\epsilon a_j(i). \]
Since the sketch allows us to recover distances up to distortion $1\pm\frac{1}{2}\epsilon$, we can recover each entry $a_j(i)$ of each $a_j$, and hence the entire sequence $a_1,\ldots,a_n$. The number of choices for this sequence is ${n\choose k}^n$, so the lower bound we get on the sketch size in bits is
\ifprocs
\[
  \log\left({n\choose k}^n\right) \geq 
  nk\log\left(\frac{n}{k}\right) 
\]
\[
  = \frac{n}{\epsilon^2}\cdot\log(n\epsilon^2)
  = \Omega(\gamma\cdot\epsilon^{-2}n\log n),
\]
\else
\[
  \log\left({n\choose k}^n\right) \geq 
  nk\log\left(\frac{n}{k}\right) = 
  \frac{n}{\epsilon^2}\cdot\log(n\epsilon^2) =
  \Omega(\gamma\cdot\epsilon^{-2}n\log n),
\]
\fi
where the final bound is since $\log(n\epsilon^2)\geq\log(n^{2\gamma})=2\gamma\log n$.

Next we prove the lower bound $\Omega(n\log\log \Phi)$. Suppose w.l.o.g.~that $\log \Phi$ is an integer. Consider the point set $X=\{1,\ldots,n\}$. Define a map $f:X\rightarrow\R$ by setting $g(1):=0$, and for every $x\in X\setminus\{1\}$ setting $g(x):=2^{\phi(x)}$ for an arbitrary $\phi(x)\in\{1,\ldots,\log \Phi\}$. This defines a set of $(\log \Phi)^{n-1}$ one-dimensional Euclidean embeddings of $X$, each of which induces a metric contained in $\mathcal D_2(n,1,\Phi)$. We can fully recover a map $g$ from this family given a sketch with distortion better than $2$, since $D(1,x)=g(x)$ for every $x\in X$. Therefore, sketching those metrics requires at least $\log\left((\log \Phi)^{n-1}\right)=\Omega(n\log\log \Phi)$ bits.

To get the final lower bound $\Omega(\gamma\cdot\epsilon^{-2}n\log n+n\log\log \Phi)$, we augment the two metric famiilies constructed above into one. We constructed a family $\mathcal F_1$ of metrics embedded in $\R^n$, of size $|\mathcal F_1| \geq 2^{\Omega(\gamma\cdot\epsilon^{-2}n\log n)}$, and a family $\mathcal F_2$ of metrics embedded in $\R^1$, of size $|\mathcal F_2| \geq 2^{\Omega(n\log\log \Phi)}$. For every $D'\in\mathcal F_1$ and $D"\in\mathcal F_2$, we can naturally define a metric $D'\oplus D"$ embedded in $\R^{n+1}$ by embedding $D'$ in the first $n$ dimensions and $D"$ in the remaining dimension. This defines a family $\mathcal F := \{D'\oplus D" : D'\in\mathcal F_1, D"\in\mathcal F_2\}$ contained in $\mathcal D_2(2n,n+1,\Phi)$ of size $|\mathcal F_1|\cdot|\mathcal F_2|$, such that a $b$-bit sketching scheme with distortion $1+\epsilon$ can recover a metric from $\mathcal F$, and the lower bound $b=\Omega(\gamma\cdot\epsilon^{-2}n\log n+n\log\log \Phi)$ follows.
\ifprocs
\qedsymbol
\fi
\end{proof}

\begin{theorem}[general metrics]\label{thm:general_lowerbound}
\[ b_{1+\epsilon}(\mathcal D_{all}(n,\Phi)) = \Omega(n^2\log(1/\epsilon) + n\log\log \Phi) . \]
\end{theorem}
\begin{proof}
Let $\epsilon>0$ and suppose w.l.o.g.~$1/\epsilon$ is an integer. Recall we use the convention $X=\{1,\ldots,n\}$. For every $x,y\in X$, $x<y$ set $d(x,y)=1+k(x,y)\cdot\epsilon$ for an arbitrary $k(x,y)\in\{0,\ldots,1/\epsilon\}$. This actually defines a metric regardless of the choice of $k$'s: we only need to verify the triangle inequality, and it holds trivially since all pairwise distances are lower-bounded by $1$ and upper-bounded by $2$. Hence we have defined a family of $(1/\epsilon)^{{n\choose2}}$ metrics. Next observe that a sketch with distortion $(1\pm\frac{1}{2}\epsilon)$ is sufficient to fully recover a metric from this family, which proves a lower bound of $\log\left((1/\epsilon)^{{n\choose2}}\right)=\Omega(n^2\log(1/\epsilon))$ on the sketch size in bits. The other lower bound $\Omega(n\log\log \Phi)$ is by the same proof as \cref{thm:euclidean_lowerbound}.
\ifprocs
\qedsymbol
\fi
\end{proof}

\paragraph{Acknowledgments.} We thank Arturs Backurs, Sepideh Mahabadi and Ilya Razenshteyn for helpful feedback on this manuscript. This work was supported in part by the NSF, MADALGO and the Simons Foundation.

\bibliographystyle{amsalpha}
\bibliography{metriccompression}

\newcommand{\etalchar}[1]{$^{#1}$}
\providecommand{\bysame}{\leavevmode\hbox to3em{\hrulefill}\thinspace}
\providecommand{\MR}{\relax\ifhmode\unskip\space\fi MR }
\providecommand{\MRhref}[2]{%
  \href{http://www.ams.org/mathscinet-getitem?mr=#1}{#2}
}
\providecommand{\href}[2]{#2}
\begin{thebibliography}{dB{\'A}GB{\etalchar{+}}13}

\bibitem[AC09]{AilonC09}
Nir Ailon and Bernard Chazelle, \emph{The fast johnson--lindenstrauss transform
  and approximate nearest neighbors}, {SIAM} J. Comput. \textbf{39} (2009),
  no.~1, 302--322.

\bibitem[Ach03]{achlioptas2003}
Dimitris Achlioptas, \emph{Database-friendly random projections:
  Johnson-lindenstrauss with binary coins}, J. Comput. Syst. Sci. \textbf{66}
  (2003), no.~4, 671--687.

\bibitem[AI06]{AndoniI06}
Alexandr Andoni and Piotr Indyk, \emph{Near-optimal hashing algorithms for
  approximate nearest neighbor in high dimensions}, 47th Annual {IEEE}
  Symposium on Foundations of Computer Science {(FOCS} 2006), 21-24 October
  2006, Berkeley, California, USA, Proceedings, 2006, pp.~459--468.

\bibitem[AK16]{AlonK16}
Noga Alon and Bo'az Klartag, \emph{Optimal compression of approximate euclidean
  distances}, arXiv preprint arXiv:1610.00239 (2016).

\bibitem[Alo03]{alon2003jllowerbound}
Noga Alon, \emph{Problems and results in extremal combinatorics—i}, Discrete
  Mathematics \textbf{273} (2003), no.~1–3, 31 -- 53, EuroComb'01.

\bibitem[Bar96]{bartal1996probabilistic}
Yair Bartal, \emph{Probabilistic approximation of metric spaces and its
  algorithmic applications}, Foundations of Computer Science, 1996.
  Proceedings., 37th Annual Symposium on, IEEE, 1996, pp.~184--193.

\bibitem[dB{\'A}GB{\etalchar{+}}13]{deBernardo2013}
Guillermo de~Bernardo, Sandra {\'A}lvarez-Garc{\'i}a, Nieves~R. Brisaboa,
  Gonzalo Navarro, and Oscar Pedreira, \emph{Compact querieable representations
  of raster data}, pp.~96--108, Springer International Publishing, Cham, 2013.

\bibitem[GGNL{\etalchar{+}}15]{Gagie15}
Travis Gagie, Javier~I. Gonz\'{a}lez-Nova, Susana Ladra, Gonzalo Navarro, and
  Diego Seco, \emph{Faster compressed quadtrees}, Proceedings of the 2015 Data
  Compression Conference (Washington, DC, USA), DCC '15, IEEE Computer Society,
  2015, pp.~93--102.

\bibitem[HPIM12]{HarpeledIM12}
Sariel Har-Peled, Piotr Indyk, and Rajeev Motwani, \emph{Approximate nearest
  neighbor: Towards removing the curse of dimensionality}, Theory of Computing
  \textbf{8} (2012), no.~14, 321--350.

\bibitem[Hud09]{hudson2009succinct}
Beno{\^\i}t Hudson, \emph{Succinct representation of well-spaced point clouds},
  arXiv preprint arXiv:0909.3137 (2009).

\bibitem[JL84]{johnson1984extensions}
William~B Johnson and Joram Lindenstrauss, \emph{Extensions of lipschitz
  mappings into a hilbert space}, Contemporary mathematics \textbf{26} (1984),
  no.~189-206, 1--1.

\bibitem[JW13]{jayram2013optimal}
Thathachar~S Jayram and David~P Woodruff, \emph{Optimal bounds for
  johnson-lindenstrauss transforms and streaming problems with subconstant
  error}, ACM Transactions on Algorithms (TALG) \textbf{9} (2013), no.~3, 26.

\bibitem[KOR98]{KOR98}
Eyal Kushilevitz, Rafail Ostrovsky, and Yuval Rabani, \emph{Efficient search
  for approximate nearest neighbor in high dimensional spaces}, Proceedings of
  the Thirtieth Annual ACM Symposium on Theory of Computing (New York, NY,
  USA), STOC '98, ACM, 1998, pp.~614--623.

\bibitem[LN16]{LarsenN16}
Kasper~Green Larsen and Jelani Nelson, \emph{Optimality of the
  johnson-lindenstrauss lemma}, arXiv preprint arXiv:1609.02094 (2016).

\bibitem[MWY13]{molinaro2013beating}
Marco Molinaro, David~P Woodruff, and Grigory Yaroslavtsev, \emph{Beating the
  direct sum theorem in communication complexity with implications for
  sketching}, Proceedings of the Twenty-Fourth Annual ACM-SIAM Symposium on
  Discrete Algorithms, Society for Industrial and Applied Mathematics, 2013,
  pp.~1738--1756.

\bibitem[PS89]{peleg1989graph}
David Peleg and Alejandro~A Sch{\"a}ffer, \emph{Graph spanners}, Journal of
  graph theory \textbf{13} (1989), no.~1, 99--116.

\bibitem[Sam88]{samet1988overview}
Hanan Samet, \emph{An overview of quadtrees, octrees, and related hierarchical
  data structures}, Theoretical Foundations of Computer Graphics and CAD,
  Springer, 1988, pp.~51--68.

\bibitem[TZ05]{thorup2005approximate}
Mikkel Thorup and Uri Zwick, \emph{Approximate distance oracles}, Journal of
  the ACM (JACM) \textbf{52} (2005), no.~1, 1--24.

\bibitem[VM14]{VenkatM14}
Prayaag Venkat and David~M. Mount, \emph{A succinct, dynamic data structure for
  proximity queries on point sets}, Proceedings of the 26th Canadian Conference
  on Computational Geometry, {CCCG} 2014, Halifax, Nova Scotia, Canada, 2014,
  2014.

\end{thebibliography}

\appendix

\section{$\ell_1$ Metrics}\label{sec:l1_metrics}
In this section we point out that in our setting, both upper and lower bounds for Euclidean metrics apply to $\ell_1$ metrics as well. In particular,
\begin{corollary}
\ifprocs
\[ b_{1+\epsilon}(\mathcal D_1(n,d,\Phi)) = \]
\[ O(\epsilon^{-2}\log(1/\epsilon) \cdot n\log n + n\log\log \Phi) , \]
\else
\[ b_{1+\epsilon}(\mathcal D_1(n,d,\Phi)) = O(\epsilon^{-2}\log(1/\epsilon) \cdot n\log n + n\log\log \Phi) , \]
\fi
and
\[ b_{1+\epsilon}(\mathcal D_1(n,d,\Phi)) = \Omega(\epsilon^{-2}n\log n + n\log\log \Phi) . \]
\end{corollary}
\begin{proof}
The upper bound follows from \cref{cor:euclidean} since every $\ell_1$-metric is of negative type, meaning it embeds isometrically into $\ell_2^2$. Then it is enough to sketch the underlying Euclidean metric. The lower bound follows from 
\cref{thm:euclidean_lowerbound} since $\ell_2$ metrics embed isometrically into $\ell_1$.
\ifprocs
\qedsymbol
\fi
\end{proof}

\section{Grid Nets}\label{sec:gridnet}
In this section we prove \cref{lmm:gridnet}. For $x\in\R^d$ and $r>0$, denote by $\mathcal{B}^d(x,r)$ the radius-$r$ ball centered at $x$ in the $\ell_p$ norm. Let $\mathcal{G}^d_\delta$ be the uniform grid with side $\delta d^{-1/p}$ in $\R^d$. The $\delta$-net for the ball would be its intersection with the grid,
\[ \mathcal{N}^d_\delta(x,r) := \mathcal{B}^d(x,r) \cap \mathcal{G}^d_\delta . \]
Clearly, given a point in the ball, we can find a $\delta$-close point in the net in time $O(d)$, by rounding each coordinate either up or down to an integer multiple of the grid side $\delta d^{-1/p}$. It remains to show how to encode and decode points in the net to bitstrings. For clarity, we present the proof for $p=2$; the analysis for any $1\leq p\leq\infty$ goes through with only a change of constants. Denote
\[ M^d_\delta(r) :=  \lceil \left(\frac{4\sqrt{\pi}\cdot r}{\delta}\right)^d \rceil . \]
We now show that $M^d_\delta(r)$ is an upper bound on the size of the net $\mathcal{N}^d_\delta(x,r)$.

\begin{fact}\label{fct:discrete_ball}
$\sum_{i=1}^m(m^2-i^2)^k \leq \sqrt{\frac{\pi}{2k}}\cdot m^{2k+1}$.
\end{fact}
\begin{proof}
\ifprocs
\[
  \frac{1}{m^{2k+1}}\sum_{i=1}^m(m^2-i^2)^k = 
  \sum_{i=1}^m\left(1-\left(\frac{i}{m}\right)^2\right)^k\frac{1}{m} \leq
\]
\[
  \int_0^1(1-x^2)^kdx =
  \frac{\sqrt\pi}{2}\cdot\frac{\Gamma(k+1)}{\Gamma(k+1.5)} \leq
  \sqrt{\frac{\pi}{2k}}. \text{\;\;\qedsymbol}
\]
\else
\[
  \frac{1}{m^{2k+1}}\sum_{i=1}^m(m^2-i^2)^k = 
  \sum_{i=1}^m\left(1-\left(\frac{i}{m}\right)^2\right)^k\frac{1}{m} \leq
  \int_0^1(1-x^2)^kdx =
  \frac{\sqrt\pi}{2}\cdot\frac{\Gamma(k+1)}{\Gamma(k+1.5)} \leq
  \sqrt{\frac{\pi}{2k}}.
\]
\fi
\end{proof}

\begin{claim}\label{clm:sumofballs}
\[ M^d_\delta(r) \geq \sum_{i=-\lfloor r\sqrt{d}/\delta \rfloor}^{\lfloor r\sqrt{d}/\delta \rfloor}M^{d-1}_\delta\left(\sqrt{r^2 - (\tfrac{\delta}{\sqrt{d}}i)^2}\right). \]
\end{claim}
\begin{proof}
\ifprocs
\[ \sum_{i=1}^{\lfloor r\sqrt{d}/\delta \rfloor}M^{d-1}_\delta\left(\sqrt{r^2 - (\tfrac{\delta}{\sqrt{d}}i)^2}\right) \]
\[ \leq \lfloor\frac{r\sqrt{d}}{\delta}\rfloor + \sum_{i=1}^{\lfloor r\sqrt{d}/\delta \rfloor}\left(\frac{4\sqrt{\pi}\cdot \sqrt{r^2 - (\tfrac{\delta}{\sqrt{d}}i)^2}}{\delta}\right)^{d-1} \]
\[ = \lfloor\frac{r\sqrt{d}}{\delta}\rfloor + \left(\frac{4\sqrt\pi}{\sqrt d}\right)^{d-1}\sum_{i=1}^{\lfloor r\sqrt{d}/\delta \rfloor}\left( \left(\frac{r\sqrt{d}}{\delta}\right)^2 - i^2 \right)^{\frac{d-1}{2}} \]
\[ \leq \lfloor\frac{r\sqrt{d}}{\delta}\rfloor + \left(\frac{4\sqrt\pi}{\sqrt d}\right)^{d-1}\sum_{i=1}^{\lceil r\sqrt{d}/\delta \rceil}\left( \left(\lceil\frac{r\sqrt{d}}{\delta}\rceil\right)^2 - i^2 \right)^{\frac{d-1}{2}} \]
\[ \leq \lfloor\frac{r\sqrt{d}}{\delta}\rfloor + \left(\frac{4\sqrt\pi}{\sqrt d}\right)^{d-1} \cdot \sqrt{\frac{\pi}{d-1}} \cdot \left(\lceil\frac{r\sqrt{d}}{\delta}\rceil\right)^d \]
\[ \leq \frac{1}{3}\left(\frac{4\sqrt{\pi}\cdot r}{\delta}\right)^d \leq \frac{1}{3}M^d_\delta(r) , \]
where we have used \cref{fct:discrete_ball} to bound the sum. Therefore, letting $r_i:=\sqrt{r^2 - (\tfrac{\delta}{\sqrt{d}}i)^2}$,
\[
  \sum_{i=-\lfloor r\sqrt{d}/\delta \rfloor}^{\lfloor r\sqrt{d}/\delta \rfloor}M^{d-1}_\delta(r_i) = 
\]
\[
  \sum_{i=-\lfloor r\sqrt{d}/\delta \rfloor}^{-1}M^{d-1}_\delta(r_i) + M^{d-1}_\delta(r) + \sum_{i=1}^{\lfloor r\sqrt{d}/\delta \rfloor}M^{d-1}_\delta(r_i) 
\]
\[
  \leq M^d_\delta(r) . \textbf{\;\;\qedsymbol}
\]
\else
\begin{align*}
& \sum_{i=1}^{\lfloor r\sqrt{d}/\delta \rfloor}M^{d-1}_\delta\left(\sqrt{r^2 - (\tfrac{\delta}{\sqrt{d}}i)^2}\right) & \\
& \leq \lfloor\frac{r\sqrt{d}}{\delta}\rfloor + \sum_{i=1}^{\lfloor r\sqrt{d}/\delta \rfloor}\left(\frac{4\sqrt{\pi}\cdot \sqrt{r^2 - (\tfrac{\delta}{\sqrt{d}}i)^2}}{\delta}\right)^{d-1} & \\
& = \lfloor\frac{r\sqrt{d}}{\delta}\rfloor + \left(\frac{4\sqrt\pi}{\sqrt d}\right)^{d-1}\sum_{i=1}^{\lfloor r\sqrt{d}/\delta \rfloor}\left( \left(\frac{r\sqrt{d}}{\delta}\right)^2 - i^2 \right)^{\frac{d-1}{2}} & \\
& \leq \lfloor\frac{r\sqrt{d}}{\delta}\rfloor + \left(\frac{4\sqrt\pi}{\sqrt d}\right)^{d-1}\sum_{i=1}^{\lceil r\sqrt{d}/\delta \rceil}\left( \left(\lceil\frac{r\sqrt{d}}{\delta}\rceil\right)^2 - i^2 \right)^{\frac{d-1}{2}} & \\
& \leq \lfloor\frac{r\sqrt{d}}{\delta}\rfloor + \left(\frac{4\sqrt\pi}{\sqrt d}\right)^{d-1} \cdot \sqrt{\frac{\pi}{d-1}} \cdot \left(\lceil\frac{r\sqrt{d}}{\delta}\rceil\right)^d & \text{by \cref{fct:discrete_ball}} \\
& \leq \frac{1}{3}\left(\frac{4\sqrt{\pi}\cdot r}{\delta}\right)^d \leq \frac{1}{3}M^d_\delta(r) .
\end{align*}
Therefore, letting $r_i:=\sqrt{r^2 - (\tfrac{\delta}{\sqrt{d}}i)^2}$,
\[
  \sum_{i=-\lfloor r\sqrt{d}/\delta \rfloor}^{\lfloor r\sqrt{d}/\delta \rfloor}M^{d-1}_\delta(r_i) = 
  \sum_{i=-\lfloor r\sqrt{d}/\delta \rfloor}^{-1}M^{d-1}_\delta(r_i) + M^{d-1}_\delta(r) + \sum_{i=1}^{\lfloor r\sqrt{d}/\delta \rfloor}M^{d-1}_\delta(r_i) \leq M^d_\delta(r) .
\]
\fi
\end{proof}

\begin{corollary}\label{cor:gridnetsize}
$\left|\mathcal{N}^d_\delta(x,r)\right| \leq M^d_\delta(r) = O(r/\delta)^d$.
\end{corollary}
\begin{proof}
By induction on $d$. In the base case $d=1$, clearly $\mathcal{N}^1_\delta(x,r) \leq \lceil r/\delta \rceil$ and the bound holds. For $d>1$, assume w.l.o.g.~$x_1=0$, let $x_{-1}$ denote the projection of $x$ on its $d-1$ last coordinates. It is a simple observation that for any $\alpha\in[-r,r]$, the points $y\in\mathcal{B}^d_\delta(x,r)$ with $y_1=\alpha$ form a $(d-1)$-dimensional ball of radius $\sqrt{r^2-\alpha^2}$. Therefore, by grouping the points in $\mathcal{N}^d_\delta(x,r)$ by their first coordinate value, we can write the grid net as a disjoint union of grid nets in $d-1$ dimensions. More precisely, denoting $r_i:=\sqrt{r^2 - (\tfrac{\delta}{\sqrt{d}}i)^2}$,
\[
  \mathcal{N}^d_\delta(x,r) = \bigcup_{i=-\lfloor r\sqrt{d}/\delta \rfloor}^{\lfloor r\sqrt{d}/\delta \rfloor}\{(\tfrac{\delta}{\sqrt{d}}i,y):y\in\mathcal{N}^{d-1}_\delta(x_{-1},r_i) \}.
\]
Then by induction,
\ifprocs
\begin{align*}
  \left|\mathcal{N}^d_\delta(x,r)\right| &\leq
  \sum_{i=-\lfloor r\sqrt{d}/\delta \rfloor}^{\lfloor r\sqrt{d}/\delta \rfloor}\left|\mathcal{N}^{d-1}_\delta(x_{-1}, r_i)\right| \\
  &\leq \sum_{i=-\lfloor r\sqrt{d}/\delta \rfloor}^{\lfloor r\sqrt{d}/\delta \rfloor}M^{d-1}_\delta(r_i) , 
\end{align*}
\else
\[
  \left|\mathcal{N}^d_\delta(x,r)\right| \leq
  \sum_{i=-\lfloor r\sqrt{d}/\delta \rfloor}^{\lfloor r\sqrt{d}/\delta \rfloor}\left|\mathcal{N}^{d-1}_\delta(x_{-1}, r_i)\right| \leq
  \sum_{i=-\lfloor r\sqrt{d}/\delta \rfloor}^{\lfloor r\sqrt{d}/\delta \rfloor}M^{d-1}_\delta(r_i) , 
\]
\fi
and the bound follows from \cref{clm:sumofballs}.
\ifprocs
\qedsymbol
\fi
\end{proof}

\paragraph{Encoding and decoding.} We map vectors in $\mathcal{N}_\delta^d(x,r)$ to integers in the range $1,\ldots,M_\delta^d(r)$, or equivalently bitstring of length $\log(M_\delta^d(r))$, as follows. For $i=-\lfloor \frac{r\sqrt{d}}{\delta} \rfloor,...,\lfloor \frac{r\sqrt{d}}{\delta} \rfloor$, we partition the range to segments of lengths $M_\delta^{d-1}(r_i)$; \cref{clm:sumofballs} ensures the sum of segments does not exceed $M_\delta^d(r)$. We group the vectors in the net by their first coordinate value, setting $N_i=\{\eta\in\mathcal{N}_\delta^d(x,r):\eta_1=\frac{\delta}{\sqrt d}i\}$, and we map $N_i$ to the segment of length $M_\delta^{d-1}(r_i)$; \cref{cor:gridnetsize} ensures the segment is large enough. Within each segment, the mapping is defined recursively, by recalling that $N_i$ projected on the last $d-1$ coordinates is the net $\mathcal N^{d-1}_\delta(x_{-1},r_i)$.

In order to encode a given vector, we need to compute the $\frac{2r\sqrt{d}}{\delta}$ segment sizes $M_\delta^{d-1}(r_i)$, pick a segment according to the first coordinate, and recurse on the remaining coordinates. The total encoding time is hence $O(d\cdot\frac{r\sqrt{d}}{\delta})$. In order to decode a given encoding, we again need to compute the segment sizes in order to the determine the first coordinate, and then recurse on offset within the current segment. The decoding time is again $O(d^{1.5}r/\delta)$.

\end{document}